\documentclass{article}
\usepackage{ijcai24}

\usepackage{times}
\usepackage{soul}
\usepackage{url}
\usepackage[hidelinks]{hyperref}
\usepackage[utf8]{inputenc}
\usepackage[small]{caption}
\usepackage{graphicx}
\usepackage{amsmath}
\usepackage{amsthm}
\usepackage{booktabs}
\usepackage{algorithm}
\usepackage{algorithmic}
\usepackage[switch]{lineno}

\usepackage{natbib}
\usepackage{cleveref}
\usepackage{amssymb}
\usepackage{xcolor}

\newtheorem{theorem}{Theorem}[section]

\newtheorem{lemma}[theorem]{Lemma}
\newtheorem{corollary}[theorem]{Corollary}
\newtheorem{observation}[theorem]{Observation}

\theoremstyle{definition}
\newtheorem{definition}[theorem]{Definition}

\newcommand{\mms}[1]{\mathrm{MMS}^{1\text{-out-of-}#1}}
\newcommand{\minshare}{p_{\mathrm{MMS}}}
\newcommand{\E}{\mathbb{E}}
\newcommand{\tA}{\tilde{A}}
\newcommand{\tX}{\tilde{X}}
\newcommand{\cA}{\mathcal{A}}

\newcommand{\cP}{\mathcal{P}}
\newcommand{\cS}{\mathcal{S}}
\newcommand{\cX}{\mathcal{X}}
\newcommand{\ind}{\mathbf{1}}

\allowdisplaybreaks

\title{Ordinal Maximin Guarantees for Group Fair Division}

\author{
Pasin Manurangsi$^1$\And
Warut Suksompong$^2$\\
\affiliations
$^1$Google Research, Thailand\\
$^2$School of Computing, National University of Singapore, Singapore\\
}

\begin{document}

\maketitle

\begin{abstract}
We investigate fairness in the allocation of indivisible items among groups of agents using the notion of maximin share (MMS).
While previous work has shown that no nontrivial multiplicative MMS approximation can be guaranteed in this setting for general group sizes, we demonstrate that ordinal relaxations are much more useful.
For example, we show that if $n$ agents are distributed equally across $g$ groups, there exists a $1$-out-of-$k$ MMS allocation for $k = O(g\log(n/g))$, while if all but a constant number of agents are in the same group, we obtain $k = O(\log n/\log \log n)$. 
We also establish the tightness of these bounds and provide non-asymptotic results for the case of two groups.
\end{abstract}

\section{Introduction}

A fundamental problem in society is how to allocate scarce resources among interested parties.
The burgeoning field of \emph{fair division} studies how to perform the allocation fairly in scenarios ranging from dividing properties between families to distributing supplies among departments of a university or neighborhoods of a city \citep{BramsTa96,Moulin19}.
Often, the resources to be allocated are \emph{indivisible}, meaning that each item is discrete and must be allocated as a whole to one of the parties.

To argue about fairness, one must specify what it means for an allocation to be ``fair''.
A popular fairness notion when allocating indivisible items is \emph{maximin share fairness} \citep{Budish11}.
If the items are allocated among $n$ agents, the \emph{maximin share (MMS)} of an agent is defined as the largest value that the agent can guarantee for herself by dividing the items into $n$ bundles and getting the worst bundle.
While an allocation that gives every agent at least her maximin share always exists when $n = 2$, this is no longer the case if $n \ge 3$ \citep{KurokawaPrWa18}.
A large stream of work has therefore focused on guaranteeing every agent a constant fraction of her maximin share, with the current best constant being marginally above $3/4$ \citep{AkramiGa24}.

The discussion in the preceding paragraph, like the majority of the fair division literature, assumes that each recipient of a bundle of items corresponds to a single agent.
However, in many allocation scenarios, each interested party is in fact a \emph{group} of agents.
Even though the agents in each group share the same set of items, they do not necessarily share the same views on the items.
Indeed, members of the same family, university department, or city neighborhood may have varying desires for different items based on their individual preferences and needs.
As a consequence, several researchers have investigated fair division among groups of agents in the last few years \citep{ManurangsiSu17,ManurangsiSu22,GhodsiLaMo18,SegalhaleviNi19,SegalhaleviSu19,KyropoulouSuVo20}.\footnote{The line of work on \emph{consensus halving} and \emph{consensus $1/k$-division} problems (e.g., \citep{SimmonsSu03,FilosratsikasGo18,GoldbergHoIg22}) can also be viewed as fair division for groups.}

The study of maximin share fairness for groups was initiated by \citet{Suksompong18}, who showed, e.g., that for two groups with two agents each, it is possible to guarantee each agent $1/8$ of her MMS.
Unfortunately, Suksompong observed that if each group contains three agents, it may already be impossible to give all of them any positive fraction of their MMS.\footnote{Suppose there are three items and each agent has utility $1$ for two of them and $0$ for the remaining item, where the three agents in each group value distinct pairs of items.
Then, each agent's MMS is~$1$, but every allocation leaves some agent with utility~$0$.}
In light of this, it may appear that maximin share fairness is not useful for group resource allocation unless the groups are extremely small.
Nevertheless, while \emph{cardinal} approximations of MMS have been intensively examined in the literature, another type of relaxation, which has also received increasing attention, is \emph{ordinal}.
This type of relaxation allows each agent to partition the items into $k$ bundles for some parameter $k > n$---the corresponding notion is called \emph{$1$-out-of-$k$ MMS}.
For individual fair division, \citet{AkramiGaTa23} proved that ordinal MMS fairness can always be satisfied if we set $k = \lceil 4n/3\rceil$, improving upon prior results by \citet{AignerhorevSe22} and \citet{HosseiniSeSe22}.
Can ordinal approximations of MMS provide general fairness guarantees for groups as well?

\subsection{Our Results}

We consider a setting where $n = n_1 + \dots + n_g$ agents are partitioned into $g\ge 2$ groups of sizes $n_1,\dots,n_g \ge 1$, respectively.
Each agent has an additive and non-negative utility function over the set of items to be allocated.
We denote by $\minshare(n_1, \dots, n_g)$ the smallest integer $p$ such that for $g$~groups of sizes $n_1,\dots,n_g$, there always exists an allocation that gives every agent at least her $1$-out-of-$p$ MMS.

We shall prove asymptotically tight bounds on $\minshare$.
Perhaps surprisingly, we show that the answer differs between the ``unbalanced'' case where one group is much larger than all other groups and the ``balanced'' case where this condition does not hold.
The first case may be applicable, for instance, when a city has a large central neighborhood in which the majority of its residents live.
Formally, we establish the following theorems.
Note that the constant $1000$ is unimportant, as the two bounds in each theorem are of the same order when $\log(n_1 + 1) = \Theta(\log[(n_2+1)\cdots(n_g+1)])$.

\begin{theorem}[Upper bounds] \label{thm:ub-main}
Let $n_1 \geq \cdots \geq n_g$ be any positive integers.
If $\frac{\log(n_1 + 1)}{\log[(n_2+1)\cdots(n_g+1)]} \le 1000$, then 
\begin{align*}
\minshare(n_1, \dots, n_g) \leq O(\log[(n_1+1)\cdots(n_g+1)]).
\end{align*}
On the other hand, if $\frac{\log(n_1 + 1)}{\log[(n_2+1)\cdots(n_g+1)]} > 1000$, then
\begin{align*}
\minshare(n_1, \dots, n_g) \leq O\left(\frac{\log(n_1 + 1)}{ \log\left(\frac{\log(n_1 + 1)}{\log[(n_2 + 1)\cdots(n_g + 1)]}\right)}\right).
\end{align*}
\end{theorem}

\begin{theorem}[Lower bounds] \label{thm:lb-main}
Let $n_1 \geq \cdots \geq n_g$ be any positive integers and $n = n_1+\dots+n_g$.
If $\frac{\log(n_1 + 1)}{\log[(n_2+1)\cdots(n_g+1)]} \le 1000$, then 
\begin{align*}
\minshare(n_1, \dots, n_g) \geq \Omega(\log[(n_1+1)\cdots(n_g+1)]).
\end{align*}
On the other hand, if $\frac{\log(n_1 + 1)}{\log[(n_2+1)\cdots(n_g+1)]} > 1000$, suppose further that $g \leq (\log n)^{1 - \delta}$ for some constant $\delta \in (0,1)$. Then, for any sufficiently large $n$ (depending on $\delta$), we have
\begin{align*}
\minshare(n_1, \dots, n_g) \geq \Omega\left(\frac{\log(n_1 + 1)}{ \log\left(\frac{\log(n_1 + 1)}{\log[(n_2 + 1)\cdots(n_g + 1)]}\right)}\right).
\end{align*}
\end{theorem}

Observe that our upper and lower bounds match in the balanced case, regardless of the relationship between the number of groups and the number of agents---for example, this is true even when $n_1,\dots,n_g\in\Theta(1)$.
While our bounds also coincide in the unbalanced case, the corresponding lower bound relies on the condition that $g\in (\log n)^{1-\Omega(1)}$.
We remark that this is a relatively mild condition, as unbalancedness already implies that $g \in O(\log n_1) \le O(\log n)$.
Moreover, our upper bounds come with efficient randomized algorithms.
To gain more intuition on our bounds, notice that if $n_1 = n_2 = \dots = n_g = n/g$, then $\minshare(n_1, \dots, n_g) = \Theta(g\log(n/g))$, whereas if $n_2 + \dots + n_g \in \Theta(1)$, then $\minshare(n_1, \dots, n_g) = \Theta(\log n/\log\log n)$.

In addition to these asymptotic bounds, we also derive concrete non-asymptotic results on $\minshare$ for the case of two groups.
When both groups are of the same size, our upper and lower bounds differ by only a factor of $2$.
Furthermore, our bounds are exactly tight for certain group sizes---for instance, we show that $\minshare(2,1) = \minshare(2,2) = \minshare(3,1) = \minshare(4,1) = 3$ and $\minshare(4,4) =  \minshare(5,5) = 4$.

\subsection{Additional Related Work}
\label{sec:relatedwork}

Although fair division has been studied extensively for several decades, the group aspect was only considered quite recently.
\citet{ManurangsiSu17} assumed that agents' utilities are drawn from probability distributions and showed that an allocation satisfying another important fairness notion, \emph{envy-freeness}, is likely to exist when there are sufficiently many items.
\citet{KyropoulouSuVo20} and \citet{ManurangsiSu22} presented worst-case guarantees with respect to relaxations of envy-freeness.
In particular, the latter authors showed that when the number of groups is constant, the smallest $k$ such that an ``envy-free up to $k$ items'' allocation always exists is $k = \Theta(\sqrt{n})$; unlike for MMS as in our work, for envy-freeness the bound is the same no matter how the $n$ agents are distributed into groups.
\citet{SegalhaleviSu19} investigated \emph{democratic fairness}, where the objective is to provide fairness guarantees to a fraction of the agents in each group; they also considered ordinal MMS, but the corresponding results are mainly restricted to binary utilities.
\citet{SegalhaleviNi19} and \citet{SegalhaleviSu21,SegalhaleviSu23} examined group fairness for \emph{divisible} resources modeled as a cake, while \citet{GhodsiLaMo18} analyzed it in rent division.

Besides the model that we consider, group fairness has also been studied in \emph{individual} fair division \citep{Berliant92,TodoLiHu11,AleksandrovWa18,BenabbouChEl19,ConitzerFrSh19,AzizRe20}.
In this setting, each agent receives a bundle of items and there is no sharing; however, fairness is desired not only for individual agents, but also for groups.
Some of these authors assumed that the groups are fixed in advance (e.g., ethnic, gender, or socioeconomic groups), whereas others required fairness to hold across all groups that could possibly be formed.
\citet{ScarlettTeZi23} explored the interplay between individual and group fairness in this model.

As we mentioned earlier, although the majority of work on MMS concerns cardinal approximations, ordinal relaxations have also attracted growing interest among researchers.
For individual fair division, \citet{AignerhorevSe22} established the existence of a $1$-out-of-$k$ MMS allocation for $k = 2n-2$.
This was later improved to $k = \lfloor 3n/2\rfloor$ \citep{HosseiniSeSe22}, and subsequently to $k = \lceil 4n/3\rceil$ \citep{AkramiGaTa23}.
Furthermore, ordinal MMS has been investigated in the context of indivisible chore allocation \citep{HosseiniSeSe22-chores}, cake cutting \citep{ElkindSeSu21,ElkindSeSu22}, as well as land division \citep{ElkindSeSu23}.

\section{Preliminaries}
\label{sec:prelim}

We use $\log$ to denote the logarithm with base~$2$ and $\ln$ to denote the natural logarithm.
For any positive integer $z$, let $[z] := \{1,\dots,z\}$.

Let $g\ge 2$ and $n_1,\dots,n_g$ be positive integers.
Let $N$ be a set of $n = n_1 + \dots + n_g$ agents partitioned into $g$ groups, where for each $i\in[g]$, the $i$-th group is denoted by $N_i$ and contains $n_i$ agents.
Let $M = [m]$ be the set of items to be allocated.
A \emph{bundle} refers to any (possibly empty) subset of~$M$.
Each agent $a\in N$ has an additive utility function $u_a : 2^M\rightarrow \mathbb{R}_{\ge 0}$ over the items in~$M$; for a single item~$\ell$, we sometimes write $u_a(\ell)$ instead of $u_a(\{\ell\})$.
For $i\in [g]$ and $j\in[n_i]$, we refer to the $j$-th agent in $N_i$ as agent~$(i,j)$, and use $u_{i,j}$ to denote the agent's utility function.
An \emph{allocation} $\mathcal{A} = (A_1,\dots,A_g)$ is an ordered partition of the items in $M$ such that for each $i\in [g]$, bundle~$A_i$ is assigned to group~$N_i$.
For $i\in[g]$, we let $$\beta_i := \frac{\log(n_i + 1)}{\log[(n_2+1)\cdots(n_g+1)]}.$$

\subsection{Ordinal MMS}

Let $p$ be a positive integer, and denote by $\Pi_p(M)$ the set of all (unordered) partitions of $M$ into $p$ parts.
We now introduce our fairness notion of interest in this paper. 

\begin{definition}
The \emph{$1$-out-of-$p$ maximin share} of an agent $a\in N$, denoted by $\mms{p}_a$, is defined as
\begin{align*}
\max_{\{B_1,\dots,B_p\}\in \Pi_p(M)} \,\min_{i\in [p]} \, u_a(B_i).
\end{align*}
An allocation $(A_1,\dots,A_g)$ is said to be an \emph{$\mms{p}$ allocation} if every agent's utility for her group's bundle is at least her $1$-out-of-$p$ maximin share.
For an agent $a\in N$ and a set of items $S\subseteq M$, we write $\mms{p}_a(S)$ to denote agent $a$'s $1$-out-of-$p$ maximin share when restricted to $S$.
\end{definition}

For any $g\ge 2$ and positive integers $n_1,\dots,n_g$, we denote by $\minshare(n_1, \dots, n_g)$ the smallest integer $p$ such that for any $g$ groups containing $n_1,\dots,n_g$ agents respectively, there exists an $\mms{p}$ allocation.
We make some observations about $\minshare$.
First, observe that if we add extra agents to existing groups, the problem can only become more difficult, as we still need to satisfy the original agents.

\begin{observation} \label{obs:dec-monotone}
For any positive integers $n_1, \dots, n_g$ and $n'_1, \dots, n'_g$ such that $n_i\ge n'_i$ for each $i\in [g]$, we have $\minshare(n_1, \dots, n_g) \geq \minshare(n'_1, \dots, n'_g)$.
\end{observation}

Similarly, if we remove a group, the problem can only become easier, as the items originally assigned to that group can be allocated arbitrarily among the remaining groups.

\begin{observation} \label{obs:dec-remove-group}
For any positive integers $n_1, \dots, n_g$ and $i \in [g]$, we have $\minshare(n_1, \dots, n_g) \geq \minshare(n_1, \dots, n_i)$.
\end{observation}

By considering $g-1$ identical items, we also have that $p$ needs to be at least the number of groups.

\begin{observation} \label{obs:atleastgroup}
For any positive integers $n_1, \dots, n_g$, we have $\minshare(n_1, \dots, n_g) \geq g$.
\end{observation}

\subsection{Covering Designs}

When deriving lower bounds, we will make use of \emph{covering designs}, which are extensively studied in combinatorics. 

\begin{definition}[e.g.,~\citep{GordonPaKu96}]
For positive integers $m,s,t$, an \emph{$(m, s, t)$-covering design} is a collection $\{S_1, \dots, S_r\}$ of subsets of $[m]$, each of size $s$, with the property that for any subset $T \subseteq [m]$ of size at most $t$, there exists $i \in [r]$ such that $T \subseteq S_i$. 
\end{definition}

Denote by $C(m, s, t)$ the size of the smallest $(m, s, t)$-covering design.
A trivial upper bound on $C(m, s, t)$ can be obtained by observing that the collection of all subsets of size~$s$ is sufficient to cover any set of size $t$.

\begin{observation}\label{obs:cov-trivial}
For any positive integers $m \geq s \geq t$, we have $C(m, s, t) \leq \binom{m}{s}$.
\end{observation}

The following two upper bounds are known in the literature.

\begin{lemma}[{\cite[Thm.~5]{ReesStWe99}}] \label{thm:cov-enum}
For any positive integers $m, k, s, t$ such that $m \geq k + st$, it holds that $C(m, m - k, t) \leq \binom{t + \lceil k / s \rceil}{t}$.
\end{lemma}

\begin{lemma}[\citep{ErdosSp74}\footnote{See, e.g., \citep[p.~270]{GordonPaKu95}.}] \label{thm:cov-greedy}
For any positive integers $m \geq s \geq t$, we have $C(m, s, t) \leq \frac{\binom{m}{t}}{\binom{s}{t}} \cdot \left(1 + \ln\binom{s}{t}\right)$.
\end{lemma}

\subsection{Concentration Bounds}

We will also use the following concentration inequalities often referred to as \emph{multiplicative Chernoff bounds}. 

\begin{lemma} \label{thm:chernoff}
Let $X_1, \dots, X_k$ be independent random variables taking values in $[0, 1]$, and let $Z := X_1 + \dots + X_k$. 
Then, for every $\delta \geq 0$, we have
\begin{align} \label{eq:chernoff-lb}
\Pr[Z \leq (1 - \delta)\E[Z]] \leq \exp\left(\frac{-\delta^2}{2}\E[Z]\right)
\end{align}
and
\begin{align} \label{eq:chernoff-ub}
\Pr[Z \geq (1 + \delta)\E[Z]] \leq \left(\frac{e^\delta}{(1 + \delta)^{1+\delta}}\right)^{\E[Z]}.
\end{align}
\end{lemma}

\section{Upper Bounds}
\label{sec:upper}

In this section, we derive upper bounds on $\minshare$ by proving \Cref{thm:ub-main}.
We will use the following algorithm for approximating the $\mms{p}$ value.

\begin{lemma}[\citep{DeuermeyerFrLa82}] \label{thm:mms-approx}
For any agent $a\in N$ and any positive integer $p$, there is a polynomial-time algorithm that computes a value $t_a$ such that $\mms{p}_a \geq t_a \geq \frac{3}{4} \cdot \mms{p}_a$.
\end{lemma}

We will also use the following results by \citet{AignerhorevSe22}.\footnote{The first result is stated as Lemma~5.1 in their paper, and the second follows from the same proof as their Theorem~1.4.}

\begin{lemma}[\citep{AignerhorevSe22}] \label{lem:greedy-mms-lb}
Suppose that an agent~$a$ has utility at most $1$ for each item, and $u_a(M) \ge 2p$ for some positive integer $p$.
Then, $\mms{p}_a \geq 1$.
\end{lemma}

\begin{lemma}[\citep{AignerhorevSe22}] \label{thm:ub-indiv}
If $n_1 = \dots = n_g = 1$, there is a polynomial-time algorithm that, given values $(t_a)_{a\in N}$ such that $t_a \le \mms{(2n-2)}_a$ for all $a\in N$, outputs an allocation that gives each agent $a\in N$ utility at least $t_a$.
\end{lemma}

We start by proving a general upper bound of $O(\log[(n_1+1) \cdots(n_g+1)])$, which already implies the first part of \Cref{thm:ub-main}. 
The high-level idea is to allocate (roughly) half of the items randomly. 
We can show that this partial allocation already satisfies the desired MMS guarantee for most of the agents. 
For the remaining agents, we use the algorithm from \Cref{thm:ub-indiv} to individually assign to them the remaining (roughly half) of the items.

\begin{theorem} \label{thm:ub-1}
For any group sizes $n_1, \dots, n_g$, let $$p = 80\left(\lceil \log(n_1 + 1) \rceil  + \dots + \lceil \log(n_g + 1)\rceil  \right).$$ 
Then, an $\mms{p}$ allocation always exists. Furthermore, there is a randomized polynomial-time algorithm that finds such an allocation with probability at least $2/3$.
\end{theorem}

\begin{proof}
First, we run the algorithm from \Cref{thm:mms-approx} to find a value $t_a$ for each $a \in N$ such that $\frac{3}{4}\cdot \mms{p}_a(M) \leq t_a \leq \mms{p}_a(M)$. 
By scaling all utilities of $a$ by a factor of $3/(4t_a)$, we may assume henceforth that $\frac{3}{4} \leq \mms{p}_a(M) \leq 1$. 
Furthermore, we can assume that each agent's utility for each item is at most $1$.
Indeed, otherwise we may decrease it to $1$, which does not change the $\mms{p}$ value nor whether a bundle satisfies $\mms{p}$ for the agent.
Our algorithm for finding the desired allocation $(A_1, \dots, A_g)$ works as follows.

\begin{itemize}
\item Let $q_1 = \frac{40\lceil \log(n_1 + 1)\rceil}{p}, \dots, q_g = \frac{40\lceil \log(n_g + 1)\rceil}{p}$, and $q = 1 - q_1 - \cdots - q_g = \frac{1}{2}$.
\item Construct $\tA_1, \dots, \tA_g, \tA$ by assigning each item from $M$ to one of these sets independently with probabilities $q_1, \dots, q_g, q$, respectively.
\item Let $N' = \bigcup_{j \in [g]} \{a \in N_j \mid u_a(\tA_j) < 1\}$.
\item Check whether the following two conditions hold:
\begin{align} \label{eq:small-set}
|N'| \leq 2g
\end{align}
and
\begin{align} \label{eq:remaining-value-large}
u_a(\tA) \geq 8g\quad &\forall a \in N'.
\end{align}
If at least one condition does not hold, the algorithm terminates with failure.
\item Run the algorithm from \Cref{thm:ub-indiv} on the agents $N'$ (with $t_a = 1$ for all $a \in N'$) and items $\tA$. Let the output be $(\tA^a)_{a \in N'}$.
\item Finally, output $(A_1, \dots, A_g)$ defined by $$A_j := \tA_j \cup \bigcup_{a \in N_j \cap N'} \tA^a.$$
\end{itemize}

Let us first argue that if the algorithm does not fail, then the output $\mathcal{A} = (A_1, \dots, A_g)$ indeed satisfies $\mms{p}$. 
By definition of $N'$ and our assumption that $\mms{p}_a(M) \leq 1$ for all agents $a\in N$, this is already satisfied for all agents in $N\setminus N'$. 
For the agents in $N'$, note that \eqref{eq:small-set} implies that $ 2|N'| \leq 4g$. 
Furthermore, \Cref{lem:greedy-mms-lb} together with \eqref{eq:remaining-value-large} implies that $\mms{4g}_{a}(\tA) \geq 1$ for all $a \in N'$. 
Thus, the guarantee of \Cref{thm:ub-indiv} ensures that these agents also receive a bundle with utility at least $1$ in $\mathcal{A}$. 
It follows that $\mathcal{A}$ satisfies $\mms{p}$.

Next, we show that the algorithm fails with probability at most $1/3$. 
To this end, we first bound the probability that \eqref{eq:small-set} fails. 
Consider any agent $a \in N_j$ for some $j\in[g]$. 
For $\ell\in M$, let $X_\ell$ denote the random variable $\ind[\ell \in \tA_j] \cdot u_a(\ell)$. 
Note that $X := \sum_{\ell \in M} X_\ell$ is exactly equal to $u_a(\tA_j)$. 
Furthermore, we have $\E[X] = q_j \cdot u_a(M) \geq q_j \cdot p \cdot \mms{p}_a(M) \geq q_j \cdot p \cdot \frac{3}{4} = 30 \lceil \log(n_j + 1)\rceil$.
Thus, applying \Cref{thm:chernoff} (specifically, \eqref{eq:chernoff-lb}), we get
\begin{align*}
\Pr[X < 1] &\leq \Pr[X \leq 0.1 \E[X]] \\
&\leq \exp(-0.4 \E[X]) \\
&\leq \exp(-12 \lceil \log(n_j + 1)\rceil) 
\leq \frac{1}{10 n_j}.
\end{align*}
The above inequality is equivalent to
\begin{align*}
\Pr[a \in N'] \leq \frac{1}{10 n_j}.
\end{align*}
As a result, we have
\begin{align*}
\E[|N'|] = \sum_{j \in [g]} \sum_{a \in N_j} \Pr[a \in N'] \leq \sum_{j \in [g]} \sum_{a \in N_j} \frac{1}{10 n_j} = \frac{g}{10}.
\end{align*}
Using Markov's inequality, we get
\begin{align} \label{eq:small-set-pr}
\Pr[|N'| > 2g] \leq \frac{1}{20}.
\end{align}

Next, we bound the probability that \eqref{eq:remaining-value-large} fails. 
Consider any agent $a \in N_j$ for some $j\in[g]$. 
For $\ell \in M$, let $\tX_\ell$ denote the random variable $\ind[\ell \in \tA] \cdot u_a(\ell)$. 
Note that $\tX := \sum_{\ell \in M} \tX_\ell$ is exactly equal to $u_a(\tA)$. 
Furthermore, we have $\E[\tX] = q\cdot u_a(M) \ge q \cdot p \cdot \mms{p}_a(M) \geq q \cdot p \cdot \frac{3}{4} = \frac{3}{8} p \geq 30 g$. 
Thus, applying \Cref{thm:chernoff} (specifically, \eqref{eq:chernoff-lb}), we get
\begin{align*}
\Pr[\tX < 8g] 
&\leq \Pr[\tX \leq 0.3 \E[\tX]] \\
&\leq \exp(-0.24 \E[\tX]) \\
&\leq \exp(-0.09 p) \\
&\leq \exp(-7.2\log[(n_1+1) \cdots (n_g+1)]) \\
&\leq \frac{1}{10(n_1+1) \cdots (n_g+1)} \\
&\leq \frac{1}{10(n_1 + \cdots + n_g)}.
\end{align*}
The above inequality is equivalent to
\begin{align*}
\Pr[u_a(\tA) < 8g] \leq \frac{1}{10(n_1 + \cdots + n_g)}.
\end{align*}
Using the union bound over all agents, we get
\begin{align} \label{eq:remaining-value-large-pr}
\Pr[\exists a \in N, u_a(\tA) < 8g] \leq \frac{1}{10}.
\end{align}

Finally, by combining \eqref{eq:small-set-pr} and \eqref{eq:remaining-value-large-pr} via the union bound, we conclude that \eqref{eq:small-set} and \eqref{eq:remaining-value-large} hold simultaneously with probability at least $1 - 1/20 - 1/10 > 2/3$.
In other words, the algorithm succeeds with probability at least $2/3$, as desired.
\end{proof}

To prove the second part of \Cref{thm:ub-main}, we again use random assignment in the first stage; however, this time we assign almost all of the items to the first group.
We then use the algorithm from \Cref{thm:ub-1} to assign the leftover items to the remaining groups.
By appealing to the stronger concentration for the upper tail bound when the required value is much larger than its expectation (i.e., \eqref{eq:chernoff-ub} in \Cref{thm:chernoff}), we arrive at the improved bound stated in the following theorem.
Note that this bound implies the desired bound in \Cref{thm:ub-main} due to the assumption that $\beta_1 > 1000$.
(Recall the definition of $\beta_1$ from \Cref{sec:prelim}.)

\begin{theorem} \label{thm:ub-2}
Let $n_1, \dots, n_g$ be group sizes such that $\beta_1 > 1000$, and let $$p = 320000\left(\left\lceil\frac{\log(n_1 + 1)}{ \log\ \beta_1}\right\rceil + \sum_{i=2}^g\lceil \log(n_i + 1) \rceil  \right).$$
Then, an $\mms{p}$ allocation always exists. 
Furthermore, there is a randomized polynomial-time algorithm that finds such an allocation with probability at least $2/3$.
\end{theorem}

\begin{proof}
Similarly to the proof of \Cref{thm:ub-1}, we may assume for every agent $a\in N$ that $\mms{p}_a(M) \in [3/4, 1]$ and that $u_a(\ell) \leq 1$ for every item $\ell \in M$.
Our algorithm for finding the desired allocation $(A_1, \dots, A_g)$ works as follows.
\begin{itemize}
\item Let $q = \frac{320\sum_{i=2}^g\lceil \log(n_i + 1) \rceil }{p}$ and $q_1 = 1 - q$.
\item Construct $\tA_1, \tA$ by assigning each item from $M$ to one of these sets independently with probabilities $q_1, q$, respectively.
\item Check whether the following two conditions hold:
\begin{align} \label{eq:first-group-mms}
u_a(\tA_1) \geq 1 \quad &\forall a \in N_1
\end{align}
and
\begin{align} \label{eq:remaining-groups-mms}
u_a(\tA) \geq qp/2 \quad &\forall a \in N_2 \cup \cdots \cup N_g.
\end{align}
If at least one condition does not hold, the algorithm terminates with failure.
\item Run the algorithm from \Cref{thm:ub-1} ten times on the $g - 1$ groups $N_2, \dots, N_g$ and items $\tA$. 
If any of the runs succeeds, let $(\tA_2, \dots, \tA_g)$ be the output of a successful run. 
Otherwise, the algorithm terminates with failure.
\item Finally, output $(\tA_1, \dots, \tA_g)$.
\end{itemize}
Observe that \eqref{eq:first-group-mms} implies that every agent in $N_1$ receives at least her $\mms{p}(M)$. 
Furthermore, \Cref{lem:greedy-mms-lb} together with \eqref{eq:remaining-groups-mms} ensures that the $\mms{qp/4}(\tA)$ value of every agent in $N_2 \cup \cdots \cup N_g$ is at least $1$. 
As a result, if the algorithm succeeds, then its output satisfies $\mms{p}$.

Next, we argue that the algorithm succeeds with probability at least $2/3$.
Since we run the algorithm from \Cref{thm:ub-1} ten times, with probability at least $1 - (1/3)^{10}$, at least one run succeeds. 
As such, it suffices to show that with probability at least $2/3 + (1/3)^{10} < 0.7$, both \eqref{eq:first-group-mms} and \eqref{eq:remaining-groups-mms} hold.

We first bound the probability that \eqref{eq:first-group-mms} holds. 
Consider any agent $a \in N_1$.
For $\ell \in M$, let $X_\ell$ denote the random variable $\ind[\ell \in \tA] \cdot u_a(\ell)$.
Note that $X := \sum_{\ell \in M} X_\ell$ is exactly equal to $u_a(\tA) = u_a(M) - u_a(\tA_1) \geq 3p/4 - u_a(\tA_1)$, where the inequality follows from the assumption that $\mms{p}_a(M) \ge 3/4$. 
Thus, applying \Cref{thm:chernoff} (specifically, \eqref{eq:chernoff-ub}) with $\delta = \frac{3p/4 - 1}{\E[X]} - 1$, we get\footnote{A later sequence of arguments we make in this proof implies that $\frac{e\cdot\E[X]}{3p/4 - 1} < 1$, so we have $\delta > 0$.}
\begin{align*}
\Pr\left[X \geq \frac{3p}{4} - 1\right] &= \Pr[X \geq (1 + \delta)\E[X]] \\
&\leq \left(\frac{e^\delta}{(1 + \delta)^{1+\delta}}\right)^{\E[X]} \\
&\leq \left(\frac{e}{1 + \delta}\right)^{(1 + \delta)\E[X]} \\
&= \left(\frac{e \cdot \E[X]}{\frac{3p}{4}-1}\right)^{\frac{3p}{4}-1}.
\end{align*}
Since $p \ge 320000\log(n_1+1)/\log\beta_1$, we have 
\[
q\le 640\sum_{i=2}^g\frac{\log(n_i+1)}{p} \le \frac{1}{500}\log\beta_1\cdot\frac{1}{\beta_1}.
\]
Observe that $\E[X] = q \cdot u_a(M)$. 
Since $\mms{p}_a(M) \leq 1$, \Cref{lem:greedy-mms-lb} implies\footnote{Otherwise, if $u_a(M) \geq 4p$, then $\mms{p}_a(M) \geq 2 \cdot \mms{2p}_a(M) \geq 2$, a contradiction.} that $u_a(M) < 4p$. 
Combining the previous two sentences, we get $\E[X] \leq 4qp$. 
Plugging this into the expression above, we thus have
\begin{align*}
\Pr\left[X \geq \frac{3p}{4} - 1\right] &\leq \left(\frac{4e q p}{\frac{3p}{4}-1}\right)^{\frac{3p}{4}-1} \\
&\leq \left(20q\right)^{\frac{3p}{4}-1} \\
&\leq \left(\frac{1}{25} \log\beta_1 \cdot \frac{1}{\beta_1}\right)^{\frac{3p}{4}-1} \\
&\overset{(\star)}{\leq} \left(\frac{1}{\beta_1}\right)^{\frac{1}{2}\left({\frac{3p}{4}-1}\right)} \\
&\leq \left(\frac{1}{\beta_1}\right)^{10 \cdot \frac{\log(n_1 + 1)}{\log\beta_1}} \\
&= \frac{1}{(n_1 + 1)^{10}} 
\leq \frac{1}{10 n_1},
\end{align*}
where for $(\star)$ we use the fact that $\frac{1}{25}\cdot \frac{\log x}{x} \leq \frac{1}{\sqrt{x}}$ for all $x \geq 1$.
The above inequality implies that $\Pr[u_a(\tA_1) < 1] \leq \frac{1}{10n_1}$. 
Taking the union bound over all $a \in N_1$, we get
\begin{align} \label{eq:first-group-mms-pr}
\Pr[\exists a \in N_1, u_a(\tA_1) < 1] \leq 0.1.
\end{align}

Next, we bound the probability that \eqref{eq:remaining-groups-mms} holds. 
Consider any agent $a \in N_2 \cup \cdots \cup N_g$. 
For $\ell \in M$, let $\tX_\ell$ denote the random variable $\ind[\ell \in \tA] \cdot u_a(\ell)$. 
Note that $\tX := \sum_{\ell \in M} \tX_\ell$ is exactly equal to $u_a(\tA)$. 
Furthermore, we have $\E[\tX] = q \cdot u_a(M) \geq q \cdot p \cdot \mms{p}_a(M) \geq \frac{3}{4}qp$. 
Thus, applying \Cref{thm:chernoff} (specifically, \eqref{eq:chernoff-lb}), we get
\begin{align*}
\Pr\left[\tX < \frac{qp}{2}\right] &\leq \Pr\left[\tX \leq \frac{2}{3} \E[\tX]\right] \\
&\leq \exp\left(-\frac{\E[\tX]}{18}\right) \\
&\leq \exp\left(-\frac{qp}{24}\right) \\
&\leq \exp(-10\log[(n_2+1) \cdots (n_g+1)]) \\
&\leq \frac{1}{10(n_2+1) \cdots (n_g+1)} \\
&\leq \frac{1}{10(n_2 + \cdots + n_g)}.
\end{align*}
The above inequality is equivalent to 
\begin{align*}
\Pr\left[u_a(\tA) < \frac{qp}{2}\right] \leq \frac{1}{10(n_2 + \cdots + n_g)}.
\end{align*}
Using the union bound over all agents $a \in N_2 \cup \cdots \cup N_g$, we get
\begin{align} \label{eq:remaining-groups-mms-pr}
\Pr\left[\exists a \in N_2 \cup \cdots \cup N_g, u_a(\tA) < \frac{qp}{2}\right] \leq 0.1.
\end{align}

Finally, by combining \eqref{eq:first-group-mms-pr} and \eqref{eq:remaining-groups-mms-pr} via the union bound, we conclude that \eqref{eq:first-group-mms} and \eqref{eq:remaining-groups-mms} hold simultaneously with probability at least $1 - 0.1 - 0.1 > 0.7$. 
In other words, the algorithm succeeds with probability at least $2/3$, as desired.    
\end{proof}

\section{Lower Bounds}
\label{sec:lower}

In this section, we turn our attention to lower bounds and prove \Cref{thm:lb-main}.
We first present a generic construction based on covering designs, and later apply this construction with more specific parameters to derive our theorem.

\begin{lemma} \label{thm:generic-lb}
Let $m, p, t_1, \dots, t_g, n_1, \dots, n_g$ be positive integers such that
\begin{enumerate}
\item $t_1 + \dots + t_g + g > m$;
\item For each $i \in [g]$, at least one of the following holds:
\begin{itemize}
\item $t_i \leq \lfloor m / p \rfloor - 1$;
\item $m - (p - 1) > t_i$ and $C(m, m - (p - 1), t_i) \leq n_i$.
\end{itemize}
\end{enumerate}
Then, $\minshare(n_1, \dots, n_g) > p$.
\end{lemma}

\begin{proof}
For each $i\in[g]$, we define the utilities of the agents in the $i$-th group $N_i$ for the $m$ items as follows.
\begin{itemize}
\item If $t_i \leq \lfloor m / p \rfloor - 1$, let each item have utility $1$ for every agent in $N_i$.
\item If $m - (p - 1) > t_i$ and $C(m, m - (p - 1), t_i) \leq n_i$, then let $\{S_{i, 1}, \dots, S_{i, n_i}\}$ be any $(m, m - (p - 1), t_i)$-covering design. 
For each $j\in[n_i]$, we define the utility function of agent $(i,j)$ by
\begin{align*}
u_{i, j}(\ell) =
\begin{cases}
\frac{1}{m - (p - 1)} &\text{ if } \ell \in S_{i, j}; \\
1 &\text{ otherwise}.
\end{cases}
\end{align*}
\end{itemize}
If both conditions are met, we pick one arbitrarily.

Suppose for contradiction that an $\mms{p}$ allocation exists.
We claim that each group $N_i$ must receive at least $t_i + 1$ items.
In the first case, group~$i$ must receive at least $\lfloor m/p\rfloor \ge t_i + 1$ items.
Consider the second case.
Each agent $(i,j)$'s $\mms{p}$ is~$1$, by putting all items from $S_{i,j}$ in one bundle and each remaining item in its own bundle.
Now, for any set $T$ of at most $t_i$ items, there exists $j$ such that $T\subseteq S_{i,j}$, and so $u_{i,j}(T) \le \frac{t_i}{m-(p-1)} < 1$.
Hence, group~$i$ must again receive at least $t_i + 1$ items in this case.
The total number of items that need to be allocated is therefore at least $\sum_{i\in [g]}(t_i + 1) = t_1 + \dots + t_g + g > m$, a contradiction.
It follows that no $\mms{p}$ allocation exists in this instance, and so $\minshare(n_1, \dots, n_g) > p$.
\end{proof}

We now apply \Cref{thm:generic-lb} using various parameter settings, starting with the balanced case.

\begin{theorem} \label{thm:lb-1}
Let $n_1\ge \dots\ge n_g\ge 63$ be positive integers such that $\lfloor \frac{1}{6} \log(n_1 + 1) \rfloor \leq \sum_{i=2}^g\left\lfloor \frac{1}{6} \log(n_i + 1) \right\rfloor$. 
Then, for $$p = \sum_{i=1}^g\left\lfloor \frac{1}{6} \log(n_i + 1) \right\rfloor,$$ we have $\minshare(n_1, \dots, n_g) > p$.
\end{theorem}

\begin{proof}
Let $t_i = 2\left\lfloor \frac{1}{6} \log(n_i + 1) \right\rfloor \ge 2$ for all $i \in [g]$. 
Note that $p = \frac{1}{2}(t_1 + \dots + t_g)$, and let $m = 2p$. 
We verify the two conditions in \Cref{thm:generic-lb}. 
For the first condition, we have $t_1 + \dots + t_g + g = m + g > m$.

It remains to check the second condition for each $i \in [g]$.
To do so, first note that
\begin{align*}
m - (p - 1) 
= p + 1 
&> \sum_{i=1}^g\left\lfloor \frac{1}{6} \log(n_i + 1) \right\rfloor \\
&\ge 2\left\lfloor \frac{1}{6} \log(n_1 + 1) \right\rfloor\\
&\ge 2\left\lfloor \frac{1}{6} \log(n_i + 1) \right\rfloor = t_i,
\end{align*}
where the second inequality follows from our assumption in the theorem.
Let $s_i = \left\lfloor\frac{p+1}{t_i}\right\rfloor \ge 1$.
We have $s_i \ge \frac{p+1}{2t_i}$ and $m = (p-1) + (p+1) \ge (p-1)+s_it_i$.
Applying \Cref{thm:cov-enum}, we get
\begin{align*}
C(m, m - (p - 1), t_i) &\leq \binom{t_i + \big\lceil \frac{p-1}{s_i} \big\rceil}{t_i} \\
&\leq \binom{t_i + \big\lceil \frac{p-1}{(p+1)/(2t_i)} \big\rceil}{t_i} \\
&\leq \binom{t_i + 2t_i}{t_i} 
< 2^{3t_i} 
\leq n_i + 1,
\end{align*}
where the last inequality is due to the definition of $t_i$.
As a result, \Cref{thm:generic-lb} implies that $\minshare(n_1, \dots, n_g) > p$.
\end{proof}

For the unbalanced case, it will be more convenient to have two separate constructions. 
First, we consider the case where the smaller groups $N_2, \dots, N_g$ are ``not too small''. 
In this case, we can still use a covering design (from \Cref{thm:cov-greedy}) for these groups.

\begin{theorem} \label{thm:lb-2}
Let $n_1 \geq n_2 \geq \cdots \geq n_g$ be positive integers such that $\beta_1 \ge 1000$ and, for each $i \in [g]$, it holds that $n_i \geq 4(\log(n_1 + 1))^4$.
Then, for $$p = \left\lfloor \frac{\log(n_1 + 1)}{10  \log\beta_1} \right\rfloor,$$ we have $\minshare(n_1, \dots, n_g) > p$.
\end{theorem}

\begin{proof}
The statement is trivial if $p = 0$, so assume without loss of generality that $p \ge 1$.
Let $m = \left\lfloor \frac{(n_1 + 1)^{1/p}}{e} \right\rfloor \cdot p$
and $t_1 = m - p$. 
We will show later that $(n_1 + 1)^{1/p} \ge 30$, which means that $m\ge 2p$ and $t_1 \ge 1$.
Let $t_i = \lceil \beta_i \cdot p \rceil$ for each $i \in \{2, \dots, g\}$. 
We verify the two conditions in \Cref{thm:generic-lb}. 
For the first condition, we have $t_1 + \cdots + t_g + g\geq t_1 + (\beta_2 \cdot p + \cdots + \beta_g \cdot p) + g = t_1 + p + g = m + g > m$.

It remains to check the second condition for each $i\in [g]$.
We do so separately for $i = 1$ and $i\ne 1$.
For $i = 1$, it is clear that $m - (p - 1) > t_1$. 
Also, we have
\begin{align*}
C(m, m - (p - 1), t_1) 
&\leq C(m, m - p, t_1) \\
&\leq \binom{m}{m - p} \\
&= \binom{m}{p} 
< \left(\frac{e m}{p}\right)^p 
\leq n_1 + 1,
\end{align*}
where the second inequality follows from \Cref{obs:cov-trivial} and
the last inequality from our choice of $m$.

Now, fix $i\in\{2,\dots,g\}$.
It is again clear that $m - (p-1) \ge p > t_i$.
By \Cref{thm:cov-greedy}, we have
\begin{align}
C(m, m &- (p - 1), t_i) \nonumber  \\&\leq C(m, m - p, t_i) \nonumber \\
&\leq \frac{\binom{m}{t_i}}{\binom{m-p}{t_i}} \cdot \left(1 + \ln\binom{m-p}{t_i}\right) \nonumber \\
&\leq \left(\frac{m-t_i}{m-p-t_i}\right)^{t_i} \left(1 + t_i \ln m\right) \nonumber \\
&\leq \exp\left(t_i \cdot \frac{p}{m-p-t_i}\right) \left(1 + t_i \ln m\right), \label{eq:cover-before-exp}
\end{align}
where for the last inequality we use the well-known bound $1+x \le \exp(x)$, which holds for all real numbers $x$.

We now bound each term separately.
Before bounding the first term, observe that 
\begin{align*}
(n_1 + 1)^{1/p} &\geq (n_1 + 1)^{\frac{10  \log\beta_1}{\log(n_1 + 1)}} 
= \beta_1^{10} \geq 30\beta_1,
\end{align*}
where the last inequality follows from our assumption that $\beta_1 \geq 1000$.
From our choice of $m$, this implies that
\begin{align}
m \geq 10\beta_1 \cdot p.
\label{eq:m-lb}
\end{align}
We can now bound the first term as follows:
\begin{align*}
\exp\left(t_i \cdot \frac{p}{m-p-t_i}\right) &\leq \exp\left((\beta_i p + 1) \cdot \frac{p}{m-2p}\right) \\
&\overset{\eqref{eq:m-lb}}{\leq} \exp\left(\frac{(\beta_i p + 1) p}{8\beta_1 \cdot p}\right) \\
&= \exp\left(\frac{\beta_i p + 1}{8\beta_1}\right) \\
&\leq 2 \exp\left(\frac{\beta_i p}{8\beta_1}\right) \\
&\leq 2 \exp\left(\frac{\beta_i \cdot \frac{\log(n_1 + 1)}{10  \log\beta_1}}{8\beta_1}\right) \\
&= 2\exp\left(\frac{\log(n_i + 1)}{80 \log\beta_1}\right) \\
&\leq 2\exp\left(\frac{\log(n_i + 1)}{80}\right) \\
&\leq 2\sqrt{n_i},
\end{align*}
where for the last inequality we use the fact that $n_i\ge 2$, which follows from our assumption on $n_i$.

Next, we bound the second term.
We have
\begin{align*}
1 + t_i \ln m 
&\leq 1 + p \cdot \ln((n_1 + 1)^{1/p} \cdot p) \\
&= 1 + p \cdot \left(\frac{1}{p} \ln(n_1 + 1) + \ln p \right) \\
&= 1 + \ln(n_1 + 1) + p  \ln p \\
&\leq 1 + \ln(n_1 + 1) + (\ln(n_1 + 1))^2 \\
&\leq 2(\ln(n_1 + 1))^2
\leq (\log(n_1 + 1))^2
\leq \sqrt{n_i} / 2,
\end{align*}
where the last inequality follows from our assumption on $n_i$.

Plugging the estimates of both terms into \eqref{eq:cover-before-exp}, we arrive at
\begin{align*}
C(m, m - (p - 1), t_i) \leq n_i.
\end{align*}
As a result, \Cref{thm:generic-lb} implies that $\minshare(n_1, \dots, n_g) > p$.
\end{proof}

In the next construction, we do not make any assumption on $N_2,\dots,N_g$ (i.e., they could all be of size $1$), and simply use the ``identical utilities'' construction.

\begin{theorem}\label{thm:lb-3}
Let $n_1$ be a positive integer such that $\log(n_1 + 1) \geq 1000(g - 1)$. 
Then, for $$p = \left\lfloor \frac{\log(n_1 + 1)}{10  \log\left(\frac{\log(n_1 + 1)}{g - 1}\right)} \right\rfloor,$$ we have $\minshare(n_1, 1, \dots, 1) > p$.
\end{theorem}

\begin{proof}
By \Cref{obs:atleastgroup}, the statement is trivial if $p \le g$, so assume without loss of generality that $p \ge g+1$.
Let $m = \left\lceil \frac{p}{g - 1} + 1 \right\rceil \cdot p$, and note that $m > p$.
Let $t_1 = m - p$, and $t_i = \left\lceil p/(g-1) \right\rceil$ for each $i \in \{2, \dots, g\}$.
We verify the two conditions in \Cref{thm:generic-lb}. 
For the first condition, we have $t_1 + \dots + t_g + g \ge t_1 + p + g = m + g > m$.

It remains to check the second condition for each $i \in [g]$.
We do so separately for $i=1$ and $i\ne 1$.
For $i=1$, it is clear that $m-(p-1) > t_1$.
Also, we have
\begin{align*}
C(m, m - (p - 1), t_1) 
&\leq C(m, m - p, t_1) \\
&\leq \binom{m}{m - p} \\
&= \binom{m}{p} \\
&< \left(\frac{e m}{p}\right)^p \\
&= \left(e\left\lceil \frac{p}{g - 1} + 1 \right\rceil\right)^p \\
&\leq \left(\frac{10p}{g - 1} \right)^p \\
&\leq \left(\frac{\log(n_1 + 1)}{g - 1}\right)^{\frac{\log(n_1 + 1)}{ \log\left(\frac{\log(n_1 + 1)}{g - 1}\right)}} \\
&= n_1 + 1,
\end{align*}
where the second inequality follows from \Cref{obs:cov-trivial}, and the last inequality from the definition of $p$ along with our assumption that $\log(n_1+1)\ge 1000(g-1)$.

For $i\in\{2,\dots,g\}$, we simply have $\lfloor m/p \rfloor - 1 = \left\lceil \frac{p}{g - 1} + 1 \right\rceil - 1 \geq t_i$; in fact, the inequality is an equality.
As a result, \Cref{thm:generic-lb} implies that $\minshare(n_1, \dots, n_g) > p$.
\end{proof}

By appropriately combining the previous three theorems for different regimes of parameters, we arrive at \Cref{thm:lb-main}.

\begin{proof}[Proof of \Cref{thm:lb-main}]
We consider the two cases separately.

\textbf{Case 1}: $\beta_1 \le 1000$.
If $\log[(n_1+1)\cdots(n_g+1)]\le 9000g$, then \Cref{obs:atleastgroup} immediately implies the desired bound.
Assume therefore that $\log[(n_1+1)\cdots(n_g+1)] > 9000g$.
Suppose for contradiction that $n_2 < 63$.
Then $n_2,\dots,n_g < 63$ and $\log(n_1+1) > 9000g - 6g > 8000g$.
On the other hand, $\beta_1 \le 1000$ implies that $\log(n_1+1) \le 1000\cdot(6g) = 6000g$, a contradiction.
Hence, $n_2 \ge 63$.

Let $g'\in\{2,\dots,g\}$ be the largest index such that $n_{g'} \ge 63$.
Since $\beta_1 \le 1000$ and $\log[(n_1 + 1)\cdots(n_g + 1)] > 9000g$, we have $\log[(n_2 + 1)\cdots(n_g + 1)] > 8g$.
It follows that
\begin{align*}
\log[(n_2+1)&\cdots(n_{g'}+1)]\\
&\ge \log[(n_2+1)\cdots(n_g+1)] - 6g\\
&\ge \Omega(\log[(n_2+1)\cdots(n_g+1)]) \\
&\ge \Omega(\log[(n_1+1)\cdots(n_g+1)]).
\end{align*}
Applying \Cref{thm:lb-1}, we get 
\begin{align*}
\minshare&(n_2, n_2, n_3, \dots, n_{g'}) \\
&\geq \Omega(\log[(n_2+1)(n_3+1) \cdots (n_{g'}+1)]) \\
&\geq \Omega(\log[(n_1+1)(n_2+1) \cdots (n_g+1)]).
\end{align*}
Finally, combining this with Observations~\ref{obs:dec-monotone} and \ref{obs:dec-remove-group} yields the desired bound.

\textbf{Case 2}: $\beta_1 > 1000$.
This means that $n_1 + 1 > 2^{1000}(n_2+1)\cdots(n_g+1)$, which implies that $n_1 > n_2 + \dots + n_g$ and therefore $n_1 > n/2$.
Let $\delta \in (0,1)$ be the constant in the theorem statement, i.e., $g\le (\log n)^{1-\delta}$.

\underline{Case 2.1}:
$\log[(n_2+1)\cdots(n_g+1)] \le (\log(n_1 + 1))^{1-\delta/2}$.
Since $\beta_1 > 1000$, we have $\log(n_1+1) > 1000(g-1)$.
Applying \Cref{thm:lb-3} and \Cref{obs:dec-monotone}, we get
\begin{align*}
\minshare(n_1, \dots, n_g) &\geq \Omega\left(\frac{\log(n_1+1)}{\log\left(\frac{\log(n_1+1)}{g}\right)}\right) \\
&\geq \Omega\left(\frac{\log(n_1+1)}{\log\log(n_1 + 1)}\right) \\
&\geq \Omega\left(\frac{\log(n_1+1)}{\log\left(\frac{\log(n_1+1)}{\log[(n_2+1)\cdots(n_g+1)]}\right)}\right),
\end{align*}
where the last relation holds by the assumption of Case~2.1 and since $\delta$ is constant.

\underline{Case 2.2}: $\log[(n_2+1)\cdots(n_g+1)] > (\log(n_1 + 1))^{1-\delta/2}$.
Let $\tau = \frac{ (\log(n_1 + 1))^{1 - \delta/2}}{2g}$, and let $g'\in\{2,\dots,g\}$ be the largest index such that $\log(n_{g'} + 1) \geq \tau$; note that $g$ is well-defined due to the assumption of Case~2.2. 
We have 
\begin{align*}
\log[(n_2+1)&\cdots(n_{g'}+1)]\\
&\ge \log[(n_2+1)\cdots(n_{g}+1)] - \tau g \\
&\ge \frac{1}{2}\log[(n_2+1)\cdots(n_{g}+1)].
\end{align*} 
Recall that $n_1 \geq n/2$. 
Hence, for any sufficiently large $n$, we have 
\begin{align*}
n_{g'} \ge 2^{\tau} - 1 \geq 2^{\Omega((\log n)^{\delta/2})} &\geq 4(\log n)^4\\
&\geq 4(\log(n_1+1))^4,
\end{align*}
where for the second relation we use the assumption $g\le (\log n)^{1-\delta}$.
We can thus apply \Cref{thm:lb-2} to conclude that
\begin{align*}
\minshare(n_1, \dots, n_{g'}) &\geq  \Omega\left(\frac{\log(n_1+1)}{\log\left(\frac{\log(n_1+1)}{\log[(n_2+1)\cdots(n_{g'}+1)]}\right)}\right) \\
&= \Omega\left(\frac{\log(n_1+1)}{\log\left(\frac{\log(n_1+1)}{\log[(n_2+1)\cdots(n_g+1)]}\right)}\right).
\end{align*}
Applying \Cref{obs:dec-remove-group} then yields the desired bound.
\end{proof}

\section{Non-Asymptotic Results}
\label{sec:non-asymptotic}

While our bounds on $\minshare$ are already asymptotically tight, in this section, we additionally present some concrete non-asymptotic bounds. 
Our focus will be on the fundamental case of two groups.
We begin with the upper bound.

\begin{theorem} \label{thm:ub-two-group-concrete}
Let $n_1, n_2, p$ be positive integers such that
\begin{align*}
\left(1 - \frac{1}{2^{p-1}}\right)^{n_1} + \left(1 - \frac{1}{2^{p-1}}\right)^{n_2} \geq 1.
\end{align*}
Then, we have
$
\minshare(n_1, n_2) \leq p.
$
\end{theorem}

Before we prove \Cref{thm:ub-two-group-concrete}, we note that the theorem, together with Bernoulli's inequality, immediately implies the following more explicit bound.
\begin{corollary} \label{cor:ub-two-group}
For any positive integers $n_1, n_2$, it holds that $\minshare(n_1, n_2) \leq 1 + \lceil \log(n_1 + n_2) \rceil$.
\end{corollary}

For the proof of \Cref{thm:ub-two-group-concrete}, we need some additional notation.
\begin{itemize}
\item For an agent $a\in N$, we write $\cS^p_a$ to denote the collection of all subsets $S\subseteq M$ such that $u_a(S) \geq \mms{p}_a(M)$.
\item We write $\cP(S)$ to denote the power set of $S$, and $T \sim \cP(S)$ to denote a (uniformly) random subset of $S$.
\item We write $\cX(S)$ to denote the set of all allocations of the items in $S$ to two groups, and write $(T_1, T_2) \sim \cX(S)$ to denote a (uniformly) random allocation of $S$.
\end{itemize}

We start by establishing the following key lemma, which gives a lower bound on the probability that a random subset of items provides $\mms{p}$ to an agent.

\begin{lemma} \label{lem:random-single-agent}
For any agent $a\in N$, we have
\begin{align*}
\Pr_{T \sim \cP(M)}[T \in \cS^p_a] > 1 - \frac{1}{2^{p - 1}}.
\end{align*}
\end{lemma}

\begin{proof}
Fix any $a\in N$.
By definition of $\mms{p}$, there exists a partition $\{B_1,\dots,B_p\}$ of $M$ such that $u_a(B_i) \geq \mms{p}_a(M)$ for all $i\in [p]$.
Observe that $T \sim \cP(M)$ can be equivalently generated as follows:
\begin{itemize}
\item Independently sample $(B^1_i, B^2_i) \sim \cX(B_i)$ for all $i \in [p]$;
\item Independently sample $j_1, \dots, j_p \sim \{1, 2\}$;
\item Let $T = \bigcup_{i \in [p]} B^{j_i}_i$.
\end{itemize}
In other words, we have
\begin{align}
&\Pr_{T \sim \cP(M)}[T \in \cS^p_a] \nonumber\\ 
&= 
\Pr\left[u_a\left(\bigcup_{i \in [p]} B^{j_i}_i\right) \geq \mms{p}_a(M)\right] \nonumber \\
&= \E\left[\Pr\left[u_a\left(\bigcup_{i \in [p]} B^{j_i}_i\right) \geq \mms{p}_a(M)\right]\right], \label{eq:pr-grouped-expanded}
\end{align}
where in the second expression, the probability is taken over $\{(B^1_i, B^2_i) \sim \cX(B_i), \,j_i \sim \{1, 2\}\}_{i\in [p]}$, and in the third expression, the expectation is taken over $\{(B^1_i, B^2_i) \sim \cX(B_i)\}_{i\in[p]}$ and the probability over $\{j_i \sim \{1, 2\}\}_{i\in [p]}$.

Consider a fixed collection of $(B^1_1, B^2_1), \dots, (B^1_p, B^2_p)$. 
We shall lower bound the inner probability above. 
To this end, first note that the probability does not change when we swap $B^1_i$ and $B^2_i$ for some $i\in [p]$, or when we swap $(B^1_i, B^2_i)$ and $(B^1_{i'}, B^2_{i'})$ for some $i,i'\in[p]$. 
Thus, we may assume without loss of generality that the following holds:
\begin{itemize}
\item $u_a(B^1_i) \geq u_a(B^2_i)$ for all $i \in [p]$;
\item $u_a(B^2_1) \leq \cdots \leq u_a(B^2_{p})$.
\end{itemize}

We claim that for all collections of $j_1, \dots, j_p \in \{1, 2\}$ such that $(j_1, \dots, j_{p - 1}) \ne (2, \dots, 2)$, it must hold that $u_a\left(\bigcup_{i \in [p]} B^{j_i}_i\right) \geq \mms{p}_a(M)$. To see this, let $\ell \in [p - 1]$ be an index such that $j_\ell = 1$. We have
\begin{align*}
u_a\left(\bigcup_{i \in [p]} B^{j_i}_i\right) &\geq u_a(B^{j_\ell}_\ell) + u_a(B^{j_p}_p) \\
&\geq u_a(B^1_\ell) + u_a(B^2_p) \\
&\geq u_a(B^1_\ell) + u_a(B^2_\ell) 
\geq \mms{p}_a(M),
\end{align*}
proving the claim.
From this claim, it follows that
\begin{align*}
&\Pr_{j_1, \dots, j_p \sim \{1, 2\}}\left[u_a\left(\bigcup_{i \in [p]} B^{j_i}_i\right) \geq \mms{p}_a(M)\right] \\
&\geq \Pr_{j_1, \dots, j_p \sim \{1, 2\}}[(j_1, \dots, j_{p - 1}) \ne (2, \dots, 2)] = 1 - \frac{1}{2^{p - 1}}.
\end{align*}
Plugging this back into \eqref{eq:pr-grouped-expanded}, we get
\begin{align*}
\Pr_{T \sim \cP(M)}[T \in \cS^p_a]  \geq 1 - \frac{1}{2^{p - 1}}.
\end{align*}

Finally, to see that the above inequality cannot be an equality, consider the case where $B^2_1 = \cdots = B^2_p = \emptyset$. 
In this case, which occurs with positive probability, it holds that $\Pr_{j_1, \dots, j_p \sim \{1, 2\}}\left[u_a\left(\bigcup_{i \in [p]} B^{j_i}_i\right) \geq \mms{p}_a(M)\right] \ge 1 - \frac{1}{2^p} > 1 - \frac{1}{2^{p - 1}}$.
This completes the proof.
\end{proof}

\Cref{lem:random-single-agent} together with the union bound already implies \Cref{cor:ub-two-group}.
To obtain the stronger \Cref{thm:ub-two-group-concrete}, we use the following combinatorial lemma by \citet{Kleitman66} before applying the union bound.\footnote{See Proposition 6.3.1 of \citet{AlonSp00} for a statement more akin to the form we use below.}
A collection $\cS$ of subsets of $M$ is said to be \emph{monotone} if, for all $S \subseteq S' \subseteq M$ such that $S \in \cS$, we have $S' \in \cS$ as well.

\begin{lemma}[\citep{Kleitman66}] \label{lem:monotone-intersection}
Let $\cS_1, \dots, \cS_k \subseteq \cP(M)$ be monotone collections of subsets of $M$. Then,
\begin{align*}
\Pr_{S \sim \cP(M)}\left[S \in \bigcap_{i\in[k]}\cS_i\right] \geq \prod_{i\in[k]} \left(\Pr_{S \sim \cP(M)}[S \in \cS_i]\right).
\end{align*}
\end{lemma}

We are now ready to prove \Cref{thm:ub-two-group-concrete}.

\begin{proof}[Proof of \Cref{thm:ub-two-group-concrete}]
For an agent $a\in N$, we write $\cX^p_a$ to denote the set of allocations that satisfy $\mms{p}$ for $a$.
Fix an arbitrary instance with two groups $N_1,N_2$, and consider the probability that a random allocation is $\mms{p}$. 
We can write this probability as
\begin{align*}
&\Pr_{\cA = (A_1, A_2) \sim \cX(M)}[\forall a \in N,\, \cA \in \cX^p_a] \\
&= 1 - \Pr_{\cA = (A_1, A_2) \sim \cX(M)}[\exists a \in N,\, \cA \notin \cX^p_a] \\
&\geq 1 - \Pr_{\cA = (A_1, A_2) \sim \cX(M)}[\exists a \in N_1,\, \cA \notin \cX^p_a] \\
&\hspace{3.8mm}- \Pr_{\cA = (A_1, A_2) \sim \cX(M)}[\exists a \in N_2,\, \cA \notin \cX^p_a] \\
&= 1 - \Pr_{S \sim \cP(M)}[\exists a \in N_1,\, S \notin \cS^p_a] \\
&\hspace{3.8mm}- \Pr_{S \sim \cP(M)}[\exists a \in N_2,\, S \notin \cS^p_a] \\
&= \Pr_{S \sim \cP(M)}[\forall a \in N_1,\, S \in \cS^p_a] \\
&\hspace{3.8mm}+ \Pr_{S \sim \cP(M)}[\forall a \in N_2,\, S \in \cS^p_a] - 1 \\
&= \Pr_{S \sim \cP(M)}\left[S \in \bigcap_{a \in N_1} \cS^p_a\right] + \Pr_{S \sim \cP(M)}\left[S \in \bigcap_{a \in N_2} \cS^p_a\right] - 1,
\end{align*}
where the inequality follows from the union bound.

Notice that the collections $\cS^p_a$ are monotone.
Applying \Cref{lem:monotone-intersection} together with \Cref{lem:random-single-agent}, we get
\begin{align*}
&\Pr_{\cA = (A_1, A_2) \sim \cX(M)}[\forall a \in N, \,\cA \in \cX^p_a] \\
&> \left(1 - \frac{1}{2^{p - 1}}\right)^{n_1} + \left(1 - \frac{1}{2^{p - 1}}\right)^{n_2} - 1.
\end{align*}
Thus, as long as the right-hand side is non-negative, an $\mms{p}$ allocation is guaranteed to exist.
\end{proof}

Next, we turn to lower bounds.
For two groups of equal size, we establish the following bound, which is already within a factor of $2$ of the upper bound from \Cref{cor:ub-two-group}.

\begin{theorem} \label{thm:lb-two-group-concrete}
For any positive integer $n'$, it holds that $\minshare(n', n') \geq 2 + \left\lfloor \frac{\log n'}{2} \right\rfloor$.
\end{theorem}

\begin{proof}
If $n' < 4$, the bound is trivial, so we assume that $n' \geq 4$.
Let $p = 1 + \left\lfloor \frac{\log n'}{2}\right\rfloor, t = p - 1$, and $m = 2t + 1$. 
By \Cref{thm:generic-lb} with $g = 2$, $t_1 = t_2 = t$, and $n_1 = n_2 = n'$, it suffices to show that $C(m, m - (p - 1), t) \leq n'$. 
To see that this is the case, we apply \Cref{obs:cov-trivial} to get
\begin{align*}
C(m, m - (p - 1), t) &\leq \binom{m}{m - (p - 1)} \\
&= \binom{2p-1}{p} 
\leq 2^{2p-2} 
\leq n',
\end{align*}
where the last inequality holds by our choice of $p$.
\end{proof}

We summarize the upper and lower bounds for $n_1,n_2\le 5$ in Table~\ref{tab:explicit}.
The bound for $n_1 = n_2 = 1$ follows from the well-known result in the individual setting (e.g., \citep{KurokawaPrWa18}). 
The remaining upper bounds are direct consequences of \Cref{thm:ub-two-group-concrete}.
As for the lower bounds, by \Cref{obs:dec-monotone}, it suffices to show that $\minshare(2, 1) \geq 3$ and $\minshare(4, 4) \geq 4$.
The former follows from Theorem~1 of \citet{Suksompong18}.
For the latter, observe that $C(5, 3, 2) \leq 4$ due to the $(5,3,2)$-covering design $\{\{1, 2, 3\}$, $\{3, 4, 5\}$, $\{2, 4, 5\}$, $\{1, 4, 5\}\}$, and apply \Cref{thm:generic-lb} with $(m,p,g,t_1,t_2,n_1,n_2) = (5,3,2,2,2,4,4)$.

\begin{table}
\centering
\begin{tabular}{c|ccccc}
			$n_1 \downarrow$ $|$ $n_2 \rightarrow$ &   $1$   &  $2$   &  $3$   &  $4$   &  $5$      \\ 
			\hline 
			$1$        & $2$ &  &  &  &   \\
			$2$        & $3$ & $3$ &  &  &   \\ 
			$3$       & $3$ & $[3,4]$ & $[3,4]$ & \textcolor{white}{$[3,4]$} & \textcolor{white}{$[3,4]$}  \\ 
			$4$       & $3$ & $[3,4]$ & $[3,4]$ & $4$ &   \\ 
			$5$       & $[3,4]$ & $[3,4]$ & $[3,4]$ & $4$ & $4$  \\ 
\end{tabular}
\caption{Bounds on $\minshare(n_1,n_2)$ for $n_1,n_2\le 5$.}
\label{tab:explicit}
\end{table}

\section{Conclusion}

In this paper, we have explored group fairness in indivisible item allocation using the popular notion of maximin share (MMS).
We showed that ordinal relaxations of MMS, unlike their cardinal counterparts, are useful for providing fairness in this setting.
In particular, we derived asymptotically tight bounds on the ordinal MMS approximations that can be guaranteed---interestingly, these bounds differ depending on whether there is a group whose size is substantially larger than the sizes of all the remaining groups.

Besides obtaining tighter non-asymptotic bounds, an interesting direction is to establish strong guarantees with respect to both MMS and envy-freeness.
In \emph{individual} fair division, there always exists an allocation that is envy-free up to one item (EF1), and such an allocation inherently satisfies $1$-out-of-$(2n-1)$ MMS \citep[Sec.~2]{SegalhaleviSu19}.
In \emph{group} fair division, however, \citet{ManurangsiSu22} showed that for constant $g$, the least~$k$ such that the existence of an EF$k$ allocation can be ensured is $k=\Theta(\sqrt{n})$; this cannot imply $1$-out-of-$\ell$ MMS for any $\ell= o(\sqrt{n})$,\footnote{Indeed, for a given agent, the $\Theta(\sqrt{n})$ items that need to be removed from the remaining groups in order to attain envy-freeness could be highly valuable to the agent.} so the implied MMS guarantee is significantly weaker than our logarithmic guarantees.
An intriguing question is therefore whether optimal EF and MMS guarantees can be made together, or whether trade-offs are inevitable.

\section*{Acknowledgments}

This work was partially supported by the Singapore Ministry of Education under grant number MOE-T2EP20221-0001 and by an NUS Start-up Grant.
We thank the anonymous reviewers for their valuable feedback.

\bibliographystyle{named}
\bibliography{ijcai24}

\begin{thebibliography}{}

\bibitem[\protect\citeauthoryear{Aigner-Horev and Segal-Halevi}{2022}]{AignerhorevSe22}
Elad Aigner-Horev and Erel Segal-Halevi.
\newblock Envy-free matchings in bipartite graphs and their applications to fair division.
\newblock {\em Information Sciences}, 587:164--187, 2022.

\bibitem[\protect\citeauthoryear{Akrami and Garg}{2024}]{AkramiGa24}
Hannaneh Akrami and Jugal Garg.
\newblock Breaking the $3/4$ barrier for approximate maximin share.
\newblock In {\em Proceedings of the 35th ACM-SIAM Symposium on Discrete Algorithms (SODA)}, pages 74--91, 2024.

\bibitem[\protect\citeauthoryear{Akrami \bgroup \em et al.\egroup }{2023}]{AkramiGaTa23}
Hannaneh Akrami, Jugal Garg, and Setareh Taki.
\newblock Improving approximation guarantees for maximin share.
\newblock {\em CoRR}, abs/2307.12916, 2023.

\bibitem[\protect\citeauthoryear{Aleksandrov and Walsh}{2018}]{AleksandrovWa18}
Martin Aleksandrov and Toby Walsh.
\newblock Group envy freeness and group {P}areto efficiency in fair division with indivisible items.
\newblock In {\em Proceedings of the 41st German Conference on AI (KI)}, pages 57--72, 2018.

\bibitem[\protect\citeauthoryear{Alon and Spencer}{2000}]{AlonSp00}
Noga Alon and Joel~H. Spencer.
\newblock {\em The Probabilistic Method}.
\newblock Wiley, 2nd edition, 2000.

\bibitem[\protect\citeauthoryear{Aziz and Rey}{2020}]{AzizRe20}
Haris Aziz and Simon Rey.
\newblock Almost group envy-free allocation of indivisible goods and chores.
\newblock In {\em Proceedings of the 29th International Joint Conference on Artificial Intelligence (IJCAI)}, pages 39--45, 2020.

\bibitem[\protect\citeauthoryear{Benabbou \bgroup \em et al.\egroup }{2019}]{BenabbouChEl19}
Nawal Benabbou, Mithun Chakraborty, Edith Elkind, and Yair Zick.
\newblock Fairness towards groups of agents in the allocation of indivisible items.
\newblock In {\em Proceedings of the 28th International Joint Conference on Artificial Intelligence (IJCAI)}, pages 95--101, 2019.

\bibitem[\protect\citeauthoryear{Berliant \bgroup \em et al.\egroup }{1992}]{Berliant92}
Marcus Berliant, William Thomson, and Karl Dunz.
\newblock On the fair division of a heterogeneous commodity.
\newblock {\em Journal of Mathematical Economics}, 21(3):201--216, 1992.

\bibitem[\protect\citeauthoryear{Brams and Taylor}{1996}]{BramsTa96}
Steven~J. Brams and Alan~D. Taylor.
\newblock {\em Fair Division: From Cake-Cutting to Dispute Resolution}.
\newblock Cambridge University Press, 1996.

\bibitem[\protect\citeauthoryear{Budish}{2011}]{Budish11}
Eric Budish.
\newblock The combinatorial assignment problem: Approximate competitive equilibrium from equal incomes.
\newblock {\em Journal of Political Economy}, 119(6):1061--1103, 2011.

\bibitem[\protect\citeauthoryear{Conitzer \bgroup \em et al.\egroup }{2019}]{ConitzerFrSh19}
Vincent Conitzer, Rupert Freeman, Nisarg Shah, and Jennifer~Wortman Vaughan.
\newblock Group fairness for the allocation of indivisible goods.
\newblock In {\em Proceedings of the 33rd AAAI Conference on Artificial Intelligence (AAAI)}, pages 1853--1860, 2019.

\bibitem[\protect\citeauthoryear{Deuermeyer \bgroup \em et al.\egroup }{1982}]{DeuermeyerFrLa82}
Bryan~L. Deuermeyer, Donald~K. Friesen, and Michael~A. Langston.
\newblock Scheduling to maximize the minimum processor finish time in a multiprocessor system.
\newblock {\em SIAM Journal on Algebraic and Discrete Methods}, 3(2):190--196, 1982.

\bibitem[\protect\citeauthoryear{Elkind \bgroup \em et al.\egroup }{2021}]{ElkindSeSu21}
Edith Elkind, Erel Segal-Halevi, and Warut Suksompong.
\newblock Graphical cake cutting via maximin share.
\newblock In {\em Proceedings of the 30th International Joint Conference on Artificial Intelligence (IJCAI)}, pages 161--167, 2021.

\bibitem[\protect\citeauthoryear{Elkind \bgroup \em et al.\egroup }{2022}]{ElkindSeSu22}
Edith Elkind, Erel Segal-Halevi, and Warut Suksompong.
\newblock Mind the gap: Cake cutting with separation.
\newblock {\em Artificial Intelligence}, 313:103783, 2022.

\bibitem[\protect\citeauthoryear{Elkind \bgroup \em et al.\egroup }{2023}]{ElkindSeSu23}
Edith Elkind, Erel Segal-Halevi, and Warut Suksompong.
\newblock Keep your distance: Land division with separation.
\newblock {\em Computational Geometry}, 113:102006, 2023.

\bibitem[\protect\citeauthoryear{Erd\H{o}s and Spencer}{1974}]{ErdosSp74}
Paul Erd\H{o}s and Joel Spencer.
\newblock {\em Probabilistic methods in combinatorics}.
\newblock Academic Press, 1974.

\bibitem[\protect\citeauthoryear{Filos-Ratsikas and Goldberg}{2018}]{FilosratsikasGo18}
Aris Filos-Ratsikas and Paul~W. Goldberg.
\newblock Consensus halving is {PPA}-complete.
\newblock In {\em Proceedings of the 50th Annual ACM SIGACT Symposium on Theory of Computing (STOC)}, pages 51--64, 2018.

\bibitem[\protect\citeauthoryear{Filos-Ratsikas \bgroup \em et al.\egroup }{2023}]{FilosratsikasHoSo23}
Aris Filos-Ratsikas, Alexandros Hollender, Katerina Sotiraki, and Manolis Zampetakis.
\newblock Consensus-halving: Does it ever get easier?
\newblock {\em SIAM Journal on Computing}, 52(2):412--451, 2023.

\bibitem[\protect\citeauthoryear{Ghodsi \bgroup \em et al.\egroup }{2018}]{GhodsiLaMo18}
Mohammad Ghodsi, Mohamad Latifian, Arman Mohammadi, Sadra Moradian, and Masoud Seddighin.
\newblock Rent division among groups.
\newblock In {\em Proceedings of the 12th International Conference on Combinatorial Optimization and Applications (COCOA)}, pages 577--591, 2018.

\bibitem[\protect\citeauthoryear{Goldberg \bgroup \em et al.\egroup }{2022}]{GoldbergHoIg22}
Paul~W. Goldberg, Alexandros Hollender, Ayumi Igarashi, Pasin Manurangsi, and Warut Suksompong.
\newblock Consensus halving for sets of items.
\newblock {\em Mathematics of Operations Research}, 47(4):3357--3379, 2022.

\bibitem[\protect\citeauthoryear{Gordon \bgroup \em et al.\egroup }{1995}]{GordonPaKu95}
Daniel~M. Gordon, Oren Patashnik, and Greg Kuperberg.
\newblock New constructions for covering designs.
\newblock {\em Journal of Combinatorial Designs}, 3(4):269--284, 1995.

\bibitem[\protect\citeauthoryear{Gordon \bgroup \em et al.\egroup }{1996}]{GordonPaKu96}
Daniel~M. Gordon, Oren Patashnik, Greg Kuperberg, and Joel Spencer.
\newblock Asymptotically optimal covering designs.
\newblock {\em Journal of Combinatorial Theory, Series A}, 75(2):270--280, 1996.

\bibitem[\protect\citeauthoryear{Hosseini \bgroup \em et al.\egroup }{2022a}]{HosseiniSeSe22-chores}
Hadi Hosseini, Andrew Searns, and Erel Segal-Halevi.
\newblock Ordinal maximin share approximation for chores.
\newblock In {\em Proceedings of the 21st International Conference on Autonomous Agents and Multi-Agent Systems (AAMAS)}, pages 597--605, 2022.

\bibitem[\protect\citeauthoryear{Hosseini \bgroup \em et al.\egroup }{2022b}]{HosseiniSeSe22}
Hadi Hosseini, Andrew Searns, and Erel Segal-Halevi.
\newblock Ordinal maximin share approximation for goods.
\newblock {\em Journal of Artificial Intelligence Research}, 74:353--391, 2022.

\bibitem[\protect\citeauthoryear{Kleitman}{1966}]{Kleitman66}
Daniel~J. Kleitman.
\newblock Families of non-disjoint subsets.
\newblock {\em Journal of Combinatorial Theory}, 1(1):153--155, 1966.

\bibitem[\protect\citeauthoryear{Kurokawa \bgroup \em et al.\egroup }{2018}]{KurokawaPrWa18}
David Kurokawa, Ariel~D. Procaccia, and Junxing Wang.
\newblock Fair enough: Guaranteeing approximate maximin shares.
\newblock {\em Journal of the ACM}, 64(2):8:1--8:27, 2018.

\bibitem[\protect\citeauthoryear{Kyropoulou \bgroup \em et al.\egroup }{2020}]{KyropoulouSuVo20}
Maria Kyropoulou, Warut Suksompong, and Alexandros~A. Voudouris.
\newblock Almost envy-freeness in group resource allocation.
\newblock {\em Theoretical Computer Science}, 841:110--123, 2020.

\bibitem[\protect\citeauthoryear{Manurangsi and Suksompong}{2017}]{ManurangsiSu17}
Pasin Manurangsi and Warut Suksompong.
\newblock Asymptotic existence of fair divisions for groups.
\newblock {\em Mathematical Social Sciences}, 89:100--108, 2017.

\bibitem[\protect\citeauthoryear{Manurangsi and Suksompong}{2022}]{ManurangsiSu22}
Pasin Manurangsi and Warut Suksompong.
\newblock Almost envy-freeness for groups: Improved bounds via discrepancy theory.
\newblock {\em Theoretical Computer Science}, 930:179--195, 2022.

\bibitem[\protect\citeauthoryear{Moulin}{2019}]{Moulin19}
Herv\'{e} Moulin.
\newblock Fair division in the internet age.
\newblock {\em Annual Review of Economics}, 11:407--441, 2019.

\bibitem[\protect\citeauthoryear{Rees \bgroup \em et al.\egroup }{1999}]{ReesStWe99}
Rolf~S. Rees, Douglas~R. Stinson, Ruizhong Wei, and G.~H.~John {van Rees}.
\newblock An application of covering designs: Determining the maximum consistent set of shares in a threshold scheme.
\newblock {\em Ars Combinatoria}, 53, 1999.

\bibitem[\protect\citeauthoryear{Robertson and Webb}{1998}]{RobertsonWe98}
Jack Robertson and William Webb.
\newblock {\em Cake-Cutting Algorithms: Be Fair if You Can}.
\newblock Peters/CRC Press, 1998.

\bibitem[\protect\citeauthoryear{Scarlett \bgroup \em et al.\egroup }{2023}]{ScarlettTeZi23}
Jonathan Scarlett, Nicholas Teh, and Yair Zick.
\newblock For one and all: Individual and group fairness in the allocation of indivisible goods.
\newblock In {\em Proceedings of the 22nd International Conference on Autonomous Agents and Multiagent Systems (AAMAS)}, pages 2466--2468, 2023.

\bibitem[\protect\citeauthoryear{Segal-Halevi and Nitzan}{2019}]{SegalhaleviNi19}
Erel Segal-Halevi and Shmuel Nitzan.
\newblock Envy-free cake-cutting among families.
\newblock {\em Social Choice and Welfare}, 53(4):709--740, 2019.

\bibitem[\protect\citeauthoryear{Segal-Halevi and Suksompong}{2019}]{SegalhaleviSu19}
Erel Segal-Halevi and Warut Suksompong.
\newblock Democratic fair allocation of indivisible goods.
\newblock {\em Artificial Intelligence}, 277:103167, 2019.

\bibitem[\protect\citeauthoryear{Segal-Halevi and Suksompong}{2021}]{SegalhaleviSu21}
Erel Segal-Halevi and Warut Suksompong.
\newblock How to cut a cake fairly: a generalization to groups.
\newblock {\em American Mathematical Monthly}, 128(1):79--83, 2021.

\bibitem[\protect\citeauthoryear{Segal-Halevi and Suksompong}{2023}]{SegalhaleviSu23}
Erel Segal-Halevi and Warut Suksompong.
\newblock Cutting a cake fairly for groups revisited.
\newblock {\em American Mathematical Monthly}, 130(3):203--213, 2023.

\bibitem[\protect\citeauthoryear{Simmons and Su}{2003}]{SimmonsSu03}
Forest~W. Simmons and Francis~Edward Su.
\newblock Consensus-halving via theorems of {B}orsuk-{U}lam and {T}ucker.
\newblock {\em Mathematical Social Sciences}, 45(1):15--25, 2003.

\bibitem[\protect\citeauthoryear{Suksompong}{2018}]{Suksompong18}
Warut Suksompong.
\newblock Approximate maximin shares for groups of agents.
\newblock {\em Mathematical Social Sciences}, 92:40--47, 2018.

\bibitem[\protect\citeauthoryear{Todo \bgroup \em et al.\egroup }{2011}]{TodoLiHu11}
Taiki Todo, Runcong Li, Xuemei Hu, Takayuki Mouri, Atsushi Iwasaki, and Makoto Yokoo.
\newblock Generalizing envy-freeness toward group of agents.
\newblock In {\em Proceedings of the 22nd International Joint Conference on Artificial Intelligence (IJCAI)}, pages 386--392, 2011.

\end{thebibliography}

\end{document}